\documentclass[a4paper,10pt]{article}
\pdfoutput = 1
\usepackage{amsmath}
\usepackage{graphicx,psfrag,epsf}
\usepackage{enumerate}
\usepackage{natbib}

\usepackage{algorithmic}
\usepackage{algorithm}
\usepackage{epsfig,amssymb,amsmath,amsthm,xspace,rotating}
\bibpunct{[}{]}{,}{a}{}{;}
\numberwithin{equation}{section}
\numberwithin{figure}{section}
\theoremstyle{plain}
\newtheorem{theorem}{Theorem} 
\newtheorem{corollary}[theorem]{Corollary}
\newtheorem{lemma}[theorem]{Lemma}
      \numberwithin{theorem}{section}

\theoremstyle{definition}
\newtheorem{definition}{Definition}[section]
\newtheorem*{notation*}{Notation}
\newtheorem*{example*}{Example}
      
\theoremstyle{remark}
\newtheorem{remark}{Remark}[section]

\newcommand{\sco}{\ensuremath{\tS^c}}
\newcommand{\sci}{\ensuremath{\tS \setminus S }}
\newcommand{\scit}{\ensuremath{\tS \setminus T }}

\newcommand{\nn}{\nonumber}
\newcommand{\be}{\begin{equation*}}
\newcommand{\ee}{\end{equation*}}
\newcommand{\bea}{\begin{align*}}
\newcommand{\eea}{\end{align*}}
\newcommand{\R}{{\mathbb R}}
\newcommand{\tS}{\widetilde{S}}

\newcommand{\hbeta}{\ensuremath{\widehat{\beta}}}
\newcommand{\hbetaA}{\ensuremath{\widehat{\beta}^{(l)}}}
\newcommand{\hbetaB}{\ensuremath{\widehat{\beta}^{(c)}}}

\newcommand{\hbetaOT}{\ensuremath{\widehat{\beta}^{OLS, T}}}

\newcommand{\norm}[1]{\ensuremath{\left\Vert #1 \right\Vert}\xspace}
\newcommand{\abs}[1]{\ensuremath{\vert #1 \vert}\xspace}
\newcommand{\norml}[1]{\ensuremath{\left\Vert #1 \right\Vert_{1}}\xspace}

\newcommand{\seq}[2]{\ensuremath{#1, \ldots, #2}}

\allowdisplaybreaks

\title{Sparse Group Selection Through Co-Adaptive Penalties}
\author{Zhou Fang}

\begin{document}

\maketitle

\begin{abstract}
Recent work has focused on the problem of conducting linear regression when the number of covariates is very large, potentially greater than the sample size. To facilitate this, one useful tool is to assume that the model can be well approximated by a fit involving only a small number of covariates -- a so called sparsity assumption, which leads to the Lasso and other methods. In many situations, however, the covariates can be considered to be structured, in that the selection of some variables favours the selection of others -- with variables organised into groups entering or leaving the model simultaneously as a special case. This structure creates a different form of sparsity. In this paper, we suggest the Co-adaptive Lasso to fit models accommodating this form of `group sparsity'. The Co-adaptive Lasso is fast and simple to calculate, and we show that it holds theoretical advantages over the Lasso, performs well under a broad set of conclusions, and is very competitive in empirical simulations in comparison with previously suggested algorithms like the Group Lasso \citep{group} and the Adaptive Lasso \citep{huangiterated}.
\end{abstract}

\section{Introduction}

Consider the standard linear regression problem, where for known $X \in \R^{n\times p}$, $Y \in \R^n$, we assume the model
\[
 Y = X\beta + \varepsilon
\]

with $\varepsilon$ an independent noise term. We fit the model, then, by trying to estimate $\beta$. 

Many modern datasets have a high dimensionality $p$, in that they have a large number of variables -- often more than the number of observations. Under this scenario, problems become very difficult. However, in recent years, it has emerged that by introducing the concept of sparsity -- the a priori assumption that $\beta$ has only $\beta_S \neq 0$, with $|S| \ll n$ -- methods can deal with such cases with both reasonable accuracy and acceptable computational requirements.

Of particular note is the Lasso \citep{tibLASSO}, which chooses, for any value of a tuning parameter $\lambda > 0$,
\[
 \hbeta_\lambda = \arg \min_\beta \frac{1}{2} \norm{Y - X\beta}^2 + \lambda \sum_{k = 1}^{p} \vert \beta_k \vert.
\]

The Lasso has been proven to have a variety of properties, many very favourable \citep{compatibility} \citep{consistLASSO} \citep{wainwrightsharp}, and fast computational schemes have been constructed \citep{osblars} \citep{pathwisecoord}. Extensions have also been proposed -- in particular, the Adaptive Lasso, which reweights the L1 Lasso penalty according to an initial estimator, can have very good performance in some contexts \citep{huangiterated} \citep{adaptiveLASSOgeer}. Nevertheless, in some situations, the sparsity assumption on which the Lasso is based may not be sufficient, or there may simply not be enough data or too many unimportant covariates for the Lasso to perform very well.

A key case we consider in this paper is the case of group sparsity. Under group sparsity, our data is augmented with a grouping structure, and we believe that not only is the data somewhat sparse, but that the grouping structure we have provides information on the patterns of sparsity that are plausible. Typically, we believe that variables in the same group are likely to be simultaneously relevant or irrelevant. 

Such a scenario occurs commonly in a very wide set of contexts. For example, when dealing with covariates that take discrete levels, constructing a design matrix with dummy variables for each level implies that having an original covariate be irrelevant is equivalent to having all its corresponding dummy variables be simultaneously irrelevant.

Existing work has addressed this problem mostly via the Group Lasso \citep{group}, which for groups $\seq{G_1}{G_q}$ and some tuning parameter $\lambda$ is defined as the estimator

\[
\hbeta_\lambda = \arg \min_\beta \frac{1}{2} \norm{Y - X\beta}^2 + \lambda \sum_{k = 1}^{q} \Vert \beta_{G_k }\Vert.
\]

Properties of the Group Lasso have been investigated by authors such as \citet{benefitgroup}, and computational methods have been investigated by \citet{meiergroup} and \citet{rothgroup}. 

However, the Group Lasso has a number of shortcomings. Firstly, it provides no room for sparsity within groups -- variables belonging to a group are either all selected, or all unselected simultaneously. But the simultaneous presence of within and among group sparsity could indeed be preferable -- for instance, when the group specification is not completely accurate, or from subject beliefs, or when we wish a sparsely representable signal. Secondly, it deals with overlapping groups in a way that might seem unnatural -- rather than selecting sparsity patterns that are unions of groups, it chooses patterns that are the complements of such an union. Finally, the computational requirements of the Group Lasso may still be above that of the Lasso. In particular, exact path solutions through LARS schemes are available only for the Lasso due to the particular piecewise linearity of Lasso solutions, and non-sparsity of signals in the Group Lasso can require many complex L2 projection steps, slowing down computations and requiring more memory. Even in online algorithms, such as \citet{onlinegroup}, to the author's knowledge, current methods show a gap in computational speed between the L1 method and Group Lasso based calculations.

A variety of previous work has been done to attempt to solve these issues. \citet{friedmanSGL} used an additive combination of the Group Lasso and Lasso penalties, while \citet{overlapgroup} modified the group penalty to deal with overlapping groups. Many of these procedures introduce their own problems, however. In particular, several algorithms introduce additional tuning parameters, requiring multidimensional grid searches to optimise for them, and hence greatly increase the computational cost.

In this paper, we adopt a different approach. Instead of using the Group Lasso penalty, we instead modify the Adaptive Lasso to use the initial estimate and calculate weights in such a way that it takes account of the grouping effect. We call this new estimation procedure the Co-adaptive Lasso, as it is a variant of the Adaptive Lasso that shares information between estimates of the coefficients.

A few other authors have independently produced approaches similar to the Co-adaptive Lasso. \citet{groupbridge} defined a variety of non-concave grouping penalties, together with the LCD algorithm, which is similar to a repeated version of the co-adaptive procedure with a constant tuning parameter. \citet{Hlas} began from a very different rationale, and also produced a similar algorithm, though again focusing on finding local minimums for a criterion function through an iterative algorithm. More broadly, in non-overlapping group cases, the form of the Co-adaptive Lasso is similar to stopped versions of the LLA algorithm of \citet{lla}, for a suitably chosen penalty function. \citet{peopleloc} also proposed a similar scheme as the O-Lasso, albeit for a very specific and dramatically different context.

Our contribution in this paper is that we focus on finite (in particular, two-stage) procedures, and prove their performance qualities, regardless of how -- and indeed, whether -- the algorithm would converge under iteration. This allows us to avoid potential issues where good properties for the global minimum can be proven, but convergence to such a minimum cannot. We separately and sequentially choose the tuning parameter at each stage of the algorithm, thereby producing an algorithm with equivalent computational cost to the Lasso itself. We also address overlapping groups and within group sparsity.

In Section~\ref{sec:notation}, we define some notation, as well as the Co-adaptive Lasso itself. In Section~\ref{sec:mainres}, we give the main results of the paper and some broad comparisons with related algorithms. We follow in Section~\ref{sec:detail} with more detailed properties. We discuss overlapping group structures in Section~\ref{sec:overlap}. Finally, in Section~\ref{sec:simulations}, we compare the Co-adaptive Lasso to other methods in simulations, and end with a discussion of further work.

\section{Notation and Definitions}\label{sec:notation}

Let $Y \in \R^n$ be the response vector. Assume, subtracting by a constant intercept term if necessary, that $\sum_{i=1}^n Y_i = 0$. $X = \left(\seq{X^{(1)}}{X^{(p)}}\right) \in \R^{n \times p}$ is the matrix composed of covariate column vectors, which we assume also to have mean zero. Hence, $n$ is the sample size, and $p$ is the number of covariates. Let $\norm{\cdotp}_n$ denote the empirical L2 norm, and $\norm{\cdotp}$ be the standard L2 norm, with $\norml{\cdotp}$ the L1 norm. 

Note that throughout, for clarity, we use small caps Roman letters $s$ to denote scalar or vector quantities, capitals $S$ to denote sets or matrices, and script letters $\mathcal{S}$ to denote sets of sets.

In our problem, the covariates have an a priori known group structure. We denote this by 
\[\mathcal{G} = \{G_j \}_{j=1}^q,\quad \mathrm{with} \quad \bigcup_{j=1}^q G_j = \{1, \ldots, p\}, \]
which defines the membership indicators of each group. We say that $\mathcal{G}$ is non-overlapping if its elements are all disjoint.

For any subset $G \subset \{1, \ldots, p\}$, not necessarily a member of $\mathcal{G}$, we then denote by $X_G$ the columns of $X$ corresponding to the indices in $G$. Similarly, for a vector $v$, say, we denote by $v_G$ the terms of $v$ corresponding to the indices in $G$. We denote by $v^+$ and $v^-$ the maximum and minimum value of $v$ respectively. 

For any set $S$, we denote by $\mathcal{G}_{\cap S} = \{G \in \mathcal{G} : G \cap S \neq \emptyset \}$, and $\mathcal{G}^c_{\cap S} = \{G \in \mathcal{G} : G \cap S = \emptyset \}$. We write $\tS$ to be $S$ together with its in-group neighbours - that is, $\tS = \cup \mathcal{G}_{\cap S}$.

For $S = \left\{ j \in \{1, \ldots, p\} : \beta_j \neq 0 \right\}$, $S$ is group sparse if $\mathcal{G}_{\cap S}$ is a small subset of $\mathcal{G}$. We say the group structure is evenly sized if each $G \in \mathcal{G}$ have the same size, and the problem is all-in-all-out (AIAO) if $S$ can be expressed exactly as an union of a small number of sets in $\mathcal{G}$.

Recall that for a tuning parameter $\lambda > 0$, the Lasso \citep{tibLASSO} estimate is defined as 

\begin{equation}
\hbeta^{(l)}_{\lambda} = \arg \min_\beta \frac{1}{2}\norm{Y - X\beta}_n^2 + \lambda \norml{\beta}. \label{definitionlas}
\end{equation}


\begin{definition}
Let $\mu > 0$. Suppose $\hbeta^{(l)}$ is a Lasso solution for $X,Y$, for some appropriately chosen tuning parameter value $\lambda$. Then, for $\mathcal{G}$ non-overlapping, we define the Co-adaptive Lasso weights, for a given covariate $j \in G$, as
\begin{align}
 w_j &= \sqrt{\abs{G}}\norm{\hbetaA_G}^{-1}.
\end{align}

The Co-adaptive Lasso solution is then
\begin{equation}
\hbetaB_{\mu} = \arg \min_\beta \frac{1}{2}\norm{Y - X\beta}_n^2 + \mu \sum_{j=1}^p w_j \vert \beta_j\vert. \label{definitioncoad}
\end{equation}

\end{definition}

In the case of overlapping groups, a range of different weight formulations may be considered, depending on the type of overlap and signal sparsity pattern. This we will delay until Section~\ref{sec:overlap}.

If each group contains only one covariate, then the weight calculations we have here is identical to the standard implementation of the Adaptive Lasso \citep{adaptiveLASSO}. 

In general, we will suppress the subscripts $\lambda$ and $\mu$ in our notation for $\hbeta$.

\subsection{Restricted eigenvalues and the Lasso}

Key to the performance of the Lasso and most variants of it are conditions on the covariance matrix - in particular its restricted eigenvalue, or compatibility properties \citep{compatibility}. Broadly speaking, these properties measure the minimal eigenvalues or generalised eigenvalues of the matrix under some set of restrictions. Failure of the relevant condition implies that there exists feasible alternative solutions that give the same fitted values, thus implying the failure of the algorithm.

In the Lasso case, we define the restricted eigenvalue as,

\[\phi^2(L, S, m ) = \min_{\delta, M} \left( \frac{\norm{X \delta}^2_n}{\norm{\delta_M}^2} : M \supset S, \vert M \vert \leq m, \norml{\delta_{S^c}} \leq L \sqrt{\abs{ S }} \norm{\delta_S} \right), \]

with $\phi^2(L,S) = \phi^2(L,S,\abs{S})$.

In \citet{compatibility} and \citet{adaptiveLASSOgeer}, a slightly different form is given, but the above is equivalently applicable. A further variant is available in \citet{Bickel}. We say the $RE(L,S,m)$ condition holds if $\phi^2(L, S, m ) \geq 0$.

The usual results for the standard Lasso in this case are

\begin{lemma} Let 
\[\lambda > 2\max \abs{X'\varepsilon/n}.\] 

Then the Lasso estimate $\hbeta_\lambda$ satisfies

\[\norm{X\hbeta_\lambda - X\beta}^2_n \leq \frac{14\lambda^2 \vert S \vert}{\phi^2(6,S)},\]

\[\norm{\hbeta_{\lambda S} - \beta_S} \leq \frac{7\lambda\sqrt{\vert S \vert}}{\phi^2(6, S)}.\]

Further,

\[\norm{\hbeta_{\lambda} - \beta} \leq \frac{28 \lambda \sqrt{\vert S \vert}}{\phi^2(6, S, 2\abs{S})}.\]
\end{lemma}

A proof of this is available in Theorem 7.1 of \citet{adaptiveLASSOgeer}. Similar theorems have been proven by other authors, including \citet{Bickel}. Use of compatibility conditions \citep{compatibility} will yield an L1 and prediction error convergence result under similar conditions. The choice $L=6$ in the instances of $\phi$ can usually be replaced with other values of $L>1$, at the price of changing the constants in the bounds.

The Adaptive Lasso in general uses the same conditions. In the Group Lasso, a similar RE property is required, but which uses the group L2 norms instead of the L1 norm in the restriction, and furthermore restricts the considered sets $M$ to be combinations of groups. This second point is one of the reasons that the required conditions for the Group Lasso are somewhat weaker than for the Lasso in the case of L2 estimation.

\subsection{Conditions for the Co-adaptive Lasso}\label{subsec:conds}

For our work, we introduce an additional variant of the RE property.

\begin{definition}\label{groupre}
We define the group restricted RE statistic as
\[\phi_\mathcal{G}^2(L, S) := \min_{\delta, M} \left( \frac{\norm{X \delta}^2_n}{\norm{\delta_M}^2} : M \in \mathcal{G}, \norml{\delta_{S^c}} \leq L \sqrt{\abs{ S }} \norm{\delta_S} \right). \]
 
\end{definition}

\begin{lemma}\label{REineq}

Let $\mathcal{H} \subset \mathcal{G}_{\cap S}$ be any covering set of $S$.
We have inequalities:
\[ \abs{\mathcal{H}} \phi^2(L,S) \geq \phi_\mathcal{G}^2(L,S) \geq \min_{G \in \mathcal{G}} \phi^2(L, S, \abs{S} + \abs{G}) \geq \phi^2(L,S)/(1+L^2\abs{S}). \]
\end{lemma}

Now, the second inequality in particular can be very loose if $\abs{G}$ is large, because we consider a much restricted set of subsets. In addition, presence of a few small groups may also not matter as much as it appears above, because these groups may not coincide with directions where $\delta$ can be large without increasing greatly $\norm{X\delta}$. In particular, $\phi_\mathcal{G}^2$ is now comparable with the Group Lasso version of the RE property. We note that the form of our inequalities are distinguished from the form in \citet{adaptiveLASSOgeer} in that $\phi^2(L,S,m) > 0$ if and only if $\phi^2(L,S) > 0$.

\begin{definition} The group restricted RE properties lead to the definition of the following conditions, which are useful for our results.
%

 Condition A1:

There exists $C > 0$ such that for all sufficiently large $n$,
\[\phi^2 (3, S) > C.\]

 Condition A2:

There exists $C > 0$ such that for all sufficiently large $n$,
\[\phi^2_{\mathcal{G}}(3, S) > C. \]

%
\end{definition}

For the co-adaptive reweighting procedure to be of benefit asymptotically, we require in addition conditions on the dimensionality of the problem and level of noise relative to the size of the signal and sample size, to ensure that the initial Lasso does not do too badly.

Condition B1:

The noise $\varepsilon$ is independent normal, with variance less than $\sigma^2$. $\norm{X^{(j)}}_n$ is bounded. Without loss of generality, assume, rescaling if necessary, that the $\norm{X^{(j)}}_n$ is identically equal to 1.

Condition B2:

There exists $\gamma_1 \geq 0,\gamma_2 \geq 0 $ such that
\[
 \max_G  \frac{\sigma^2\abs{S} \log(p) \abs{G}^{\gamma_1}}{n\max(\norm{\beta_G}^2, \min_{H \in \mathcal{G}_{\cap S}} \norm{\beta_H}^2) } = o\left(n^{-\gamma_2}\right).
\]

The A and B conditions together imply certain convergences in the co-adaptive weights. A final set of conditions govern the convergence of the second stage:

Write 
\[
L_1^0 =  \max_{\substack{G \in \mathcal{G}_{\cap S} ,\\ H \in \mathcal{G}_{\cap S}^c}}\sqrt{\frac{\abs{G}}{\abs{H}^{1+\gamma_1} n^{\gamma_2}}}.
\]
Condition C1:

(a): Given condition B2, there exists $C >  0$, $\delta > 0$ such that

\[
\phi^2\left( \delta L_1^0 , \tS\right) > C.
\]

(b): Further, for some $T \supset S$,
\[
\phi^2\left( \delta L_1^0 , \tS, \abs{\tS} + \abs{T}\right) > C.
\]

Condition C2:

(a): There exists $C>0$, $\delta > 0$ such that there exists $T$, $S \subseteq T \subseteq \tS$ satisfying
\[
\phi^2\left(\max\left( \delta L_1^0, \max_{\substack{G \in \mathcal{G}_{\cap S}, \\H \in \mathcal{G}_{\cap (\scit)}}} 2(1+\delta)\frac{\norm{\beta_H}/\sqrt{\abs{H}}}{\norm{\beta_G}/\sqrt{\abs{G}}}\right) , T\right) > C.
\]
(b): Further,
\[
\phi^2\left(\max\left( \delta L_1^0 , \max_{\substack{G \in \mathcal{G}_{\cap S}, \\H \in \mathcal{G}_{\cap (\scit)}}} 2(1+\delta)\frac{\norm{\beta_H}/\sqrt{\abs{H}}}{\norm{\beta_G}/\sqrt{\abs{G}}}\right) , T, 2\abs{T}\right) > C.
\]
\begin{remark}\label{constantcoefs}
Conditions A1 and A2 are restricted eigenvalue type conditions to ensure the Lasso performs sufficiently well in the first stage. They generally hold so long as the covariates are not too highly correlated, and $\abs{S}$ and $p$ are not too large relative to $n$. Generally, these conditions are a bit weaker than those required for the Lasso.

Conditions B1 and B2 govern the noise level and scaling of the various dimensions of the problem. The normality assumption in B1 can be trivially relaxed to any general subgaussian noise distribution. B2 holds, for example, in any case where the Lasso itself converges and $\min \norm{\beta_G}$ does not decrease. In the case where, in addition, the expected sizes of the individual coefficients remain constant as $\abs{G}$ increases, B2 holds with $\gamma_1 \geq 1$, $\gamma_2 > 0$. Note that many bounds in this article may hold even if B2 fails. However, the bounds may no longer be useful.

C1 and C2 are restricted eigenvalue type conditions for the second stage. They are similar to those in A1-A2. but include allowances for, on the plus side, the first stage's success in removing irrelevant groups, and on the negative side, differences between the groups that make it difficult to identify relevant or irrelevant variables. C1 is most useful in cases where there is little within-group sparsity; while C2 allows success when the group sizes are too large for C1 to be satisfied, as long as there is substantial within-group sparsity. 

Relations exist between the conditions. The (b) parts of C1 and C2 imply their corresponding (a) parts. By Lemma~\ref{REineq}, if $\min \{\abs{\mathcal{H}} : \cap \mathcal{H} \supset S\}$ remains bounded, then A2 implies A1. Further, if $RE(3, S, \abs{S} + \max \abs{G} )$ is satisfied, then both conditions A1 and A2 are satisfied. The main theorems in this article require the satisfaction of all the A and B conditions, plus one of condition C1 or C2. 

Condition C1 simplifies in the case where the groups are evenly sized, in which case we require only that $\phi^2(L, \tS)$ is bounded for some fixed $L$ for this to eventually automatically be fulfilled. Indeed, in this evenly sized case, if $\phi^2(3, \tS, 2\abs{\tS})$ is bounded, then conditions C1 and A1-2 are satisfied, since $\max\abs{G} \leq \abs{\tS}$ . In the evenly sized groups, AIAO context, then, satisfaction of the conditions for the Lasso implies the conditions for the Co-adaptive Lasso are satisfied.

In this case of evenly sized groups, for C2, the condition becomes dependent on the ratio $\norm{\beta_H}/\norm{\beta_G}$. Noting that this ratio is always greater or equal to 1, satisfying the condition becomes a trade-off between choosing T large enough so that this ratio is kept small, and small enough so that restricted eigenvalues do not fall too low. Because of the role of $\abs{T}$ in our later theorems, condition C2 is most useful when a $T$ can be found with size $\abs{T} = O(\abs{S})$.

\end{remark}
\section{Main results}\label{sec:mainres}

For simplicity, we will focus on the case of non-overlapping groups. Overall, due to the re-weighting nature of the algorithm, the Co-adaptive Lasso necessarily inherits some of the properties of the Adaptive Lasso \citep{adaptiveLASSOgeer}. Whenever the conditions necessary for the the Adaptive Lasso are fulfilled, for example, the adaptive weights must successfully distinguish between the zero and the non-zero coefficients, with the weights on the non-zero coefficients converging to a negligible fraction of the weights on the zero coefficients. In this instance, the co-adaptive weights, as an aggregation of adaptive weights within each group, must also successfully distinguish between zero and non-zero groups, so the second Lasso step converges to a Lasso performed on the subset of covariates that belong to non-zero groups, instead of the whole $p$-dimensional dataset. In the asymptotic setting of the sample size increasing to infinity, taking $\mu \rightarrow 0$, we obtain results in the second stage similar to an ordinary linear regression conducted only on the relevant covariate groups, implying that the Co-adaptive Lasso is consistent for AIAO sparsity in the cases where the Adaptive Lasso succeeds. However, there may be situations where the Co-adaptive Lasso succeeds though the Adaptive Lasso does not, or at least attains a better performance.

The following theorem gives some asymptotic bounds:

\begin{theorem}\label{bigmain}
Suppose conditions A1-2 and B1-2 are satisfied with $\mathcal{G}$ non-overlapping. If for some $T$ with $S \subseteq T \subseteq \tS$, either C1(a) or C2(a) is satisfied, then, writing 
\begin{align*}
 \widetilde{\mu} &= \max \left(\sqrt{\log(p)} L_1^0, \sqrt{\log{\abs{\tS}}}  \mathop{\max_{G \in \mathcal{G}_{\cap S},}}_{  H \in \mathcal{G}_{\cap (\tS \setminus T)}}  \frac{\norm{\beta_H} \sqrt{\abs{G}}}{\norm{\beta_G} \sqrt{\abs{H}}}\right) \\
\widetilde{L} &= \max \left(L_1^0, \mathop{\max_{G \in \mathcal{G}_{\cap S},}}_{  H \in \mathcal{G}_{\cap (\tS \setminus T)}}  \frac{\norm{\beta_H} \sqrt{\abs{G}}}{\norm{\beta_G} \sqrt{\abs{H}}}\right),         
\end{align*}
taking the latter part of the maximisation equal to $0$ if $\scit$ is empty, we have that for any $\eta > 0$, there exists $\lambda$, $\mu$ such that

\begin{align}
\norm{X\left(\hbetaOT - \hbetaB\right)}_n^2 &= O\left(\sigma^2\frac{\abs{T}}{n} \widetilde{\mu}^2\right)\\
\norm{X\left(\beta - \hbetaB\right)}_n^2 &= O\left(\sigma^2\frac{\abs{T}}{n} \left(1+ \widetilde{\mu}\right)^2\right) \\
\end{align}
with probability exceeding $1 - \eta $. Here $\hbetaOT$ represents the ordinary least squares solution restricted to the covariate set $T$.

Further if either condition C1(b) or C2(b) holds, simultaneously
\begin{align}
\norm{\hbetaOT - \hbetaB}^2 &= O\left(\sigma^2\frac{\abs{T}}{n} \widetilde{\mu}^2\left(1 + \widetilde{L}  \right)^2\right) \\
\norm{\beta - \hbetaB}^2 &= O\left(\sigma^2\frac{\abs{T}}{n} \left(1+\widetilde{\mu}\left(1 + \widetilde{L}  \right)\right)^2\right).
\end{align}
 
\end{theorem}

In short, so long as the initial Lasso can be guaranteed to not perform too badly, making allowances for a limited set of variables $T \setminus S$ where it is impossible to distinguish relevant from irrelevant, the Co-adaptive Lasso performs within a multiplicative factor of the optimal result, with this factor being dependent on variability in group sizes and coefficient group norms.

In some common situations, Theorem~\ref{bigmain} can be simplified greatly with some additional assumptions.
\begin{corollary}\label{corpropor}
 Suppose that $\mathcal{G}$ is non-overlapping, and conditions A1-2 and B1-2 are satisfied. Suppose additionally that 
\[
\max_{H \in  \mathcal{G}_{\cap S}^c} \abs{H} = O\left( \min_{G\in  \mathcal{G}_{\cap S}} \abs{G}\right) 
\]
and there exist $T \supseteq S$ with $\abs{T} = O(\abs{S})$ so that C1(a) or C2(a) is satisfied and
\[
 \max_{H \in \mathcal{G}_{\cap (\scit)}} \norm{\beta_H}/\abs{H} = O\left( \min_{G \in \mathcal{G}_{\cap S}} \norm{\beta_G}/\abs{G} \right).
\]

Then for any $\eta >0$, there exists $\lambda, \mu$ such that for any $G$,
\[
\norm{X \left(\beta - \hbetaB\right)}_n^2 = O\left(\sigma^2 \frac{\abs{S}}{n}\max\left(\frac{\log \abs{\mathcal{G}}}{\abs{G}^{\gamma_1}n^{\gamma_2}} , \log \abs{\tS} \right)\right),
\]
with probability exceeding $1-\eta$.

If C1(b) or C2(b) is satisfied then simultaneously
\[
 \norm{\beta - \hbetaB}^2 = O\left(\sigma^2 \frac{\abs{S}}{n}\max\left(\frac{\log \abs{\mathcal{G}}}{\abs{G}^{\gamma_1}n^{\gamma_2}} , \log \abs{\tS} \right)\right).
\]
\end{corollary}

\begin{corollary}\label{coraiao}

Suppose that $\max_{H \in  \mathcal{G}_{\cap S}^c} \abs{H} = O\left( \min_{G\in  \mathcal{G}_{\cap S}} \abs{G}\right)$, and $\mathcal{G}$ is non-overlapping. Then if conditions B1-2 and A2 are satisfied and A1 satisfied replacing $S$ with $\tS$, then, for any $\eta > 0$, there exist $\lambda, \mu$ such that for any $G$,
\[
 \norm{X \left( \beta - \hbetaB\right)}^2_n = O\left( \sigma^2\frac{\abs{\tS}}{n}(1+ \sqrt{\log \abs{\mathcal{G}}\abs{G}^{-\gamma_1} n^{-\gamma_2}})^2 \right),
\]
with probability exceeding $1 - \eta$.

If in addition there exists fixed $L, C$ such that $\phi(L, \tS, 2\abs{\tS}) > C$, then it is simultaneously the case that
\[
 \norm{ \beta - \hbetaB}^2 = O\left( \sigma^2\frac{\abs{\tS}}{n}(1+ \sqrt{\log \abs{\mathcal{G}}\abs{G}^{-\gamma_1} n^{-\gamma_2}})^2 \right).
\]
\end{corollary}

In particular, Corollary~\ref{coraiao}, in the case of AIAO signals, requires conditions that are the same or weaker than those of the ordinary Lasso or Adaptive Lasso.

We stress that these bounds in general only provide worst case results. In contrast to the Lasso, where the penalisation also provides a tight lower bound on the prediction and estimation errors \citep{benefitgroup}, in realistic cases, the distribution of errors amongst the irrelevant groups means that the weight calculations can be, and usually will be much better than Theorem~\ref{bigmain} suggests. Since the Lasso selects a maximum of $\min(n,p)$ variables, a simple calculation will show that in the best case, under the conditions of Corollary~\ref{coraiao}, we spread the error in the initial estimate across $\min (n,  \abs{\mathcal{G}})$ groups, resulting in 

\[\norm{ X(\beta  - \hbetaB)  }^2_n =  O \left( \sigma^2 \frac{\abs{S}}{n} (1+ \sqrt{\log\abs{\mathcal{G}} \abs{G}^{-\gamma_1}n^{-\gamma_2}/\min(n, \abs{\mathcal{G}})  })^2  \right) ,\]

with similar results for the other inequalities.

%

Let us compare the above bounds to performance bounds for some related methods. In the following, we focus on the prediction error rate $\norm{X(\beta -\hbetaB)}_n^2$. We note that in general similar results can be obtained for the estimation error.

\subsection{Comparison to the Lasso}

Now, the standard Lasso has a prediction error on the order of
\[
\norm{X(\beta - \hbetaA)}^2_n =  O\left(\sigma^2\frac{\abs{S}}{n}\log(p)\right).
\]

Under the conditions of Corollary~\ref{coraiao}, however, the Co-adaptive Lasso has a prediction error of
\[
 \norm{X(\beta - \hbetaB)}^2_n =  O\left(\sigma^2\frac{\abs{\tS}}{n}\left(1 + \sqrt{\log \abs{\mathcal{G}}\abs{G}^{-\gamma_1} n^{-\gamma_2}}\right)^2\right).
\]
Hence, when the conditions are satisfied, the Co-adaptive Lasso can outperform the ordinary Lasso, assuming that the group or sample size is large, and the grouping structure is meaningful in that $\tS$ is kept small. Indeed, in the AIAO case, if $n \rightarrow \infty$, then the Co-adaptive Lasso attains the oracle rate, removing the contribution from $p$. Indeed, in this case, examination of the Irrepresentable Condition \citep{consistLASSO} indicates the Co-adaptive Lasso should be consistent for variable selection under much weaker design conditions than the Lasso.

If on the other hand the conditions of Corollary~\ref{corpropor} are satisfied, and $\abs{\mathcal{G}}$ does not grow exponentially in $\abs{G}n$, then 
\[
 \norm{X(\beta - \hbetaB)}^2_n = O\left( \sigma^2 \frac{\abs{S}}{n} \log \abs{\tS} \right).
\]
In effect, the co-adaptive re-weighting has successfully screened out all of the variables that do not belong in the same group as the relevant variables. Since often $\abs{\tS} = O(n)$, while p can be anything up to exponential in $n$, this can be a great improvement. Indeed, if $\abs{\tS} = O(n)$, we arrive at bounds asymptotically within a constant factor of the oracle rate. 

On the other hand, if condition B2 cannot be satisfied, or $\tS$ is too large relative to $S$, the Co-adaptive Lasso can under-perform. In empirical experiments, though, we see that the Co-adaptive Lasso often outperforms the Lasso even with randomly chosen groups - after all, for small enough groups, the Co-adaptive Lasso acts similarly to the Adaptive Lasso, which can have performance advantages.
%
%
%

\subsection{Comparison to the Adaptive Lasso}

The usefulness of Theorem~\ref{bigmain} lies in the dependence on $(\min_{G  \in \mathcal{G}_{\cap S}} \norm{\beta_G})$ implied in condition B2. In contrast, an Adaptive Lasso approach, being equivalent to a Co-adaptive Lasso with group size 1, would use $\min \abs{\beta_S}$ instead. Suppose that the average squared $\beta$ amongst the true $S$ remains constant or at least bounded below. Then as group sizes increase, $(\min_{G  \in \mathcal{G}_{\cap S}} \norm{\beta_G}) = \Omega(\min_G \sqrt{\abs{G}})$, satisfying condition B2 with $\gamma_1 = 1$. Meanwhile, in many set-ups, the minimum $\min \abs{\beta_S}$ would not increase, but indeed decrease, and hence if the initial Lasso gives errors that are larger than this, performance guarantees cannot be given. In terms of $\gamma_2$, if the Adaptive Lasso satisfies B2 for any particular $\gamma_2$, it is implied that the Co-adaptive Lasso must also satisfy it for that $\gamma_2$. Similar results arise if we focus on a harmonic mean based formulation that gives slightly tighter bounds.

On the other hand, if there is too much within group sparsity, the Adaptive Lasso may be superior, because the Co-adaptive Lasso fails to discriminate as strongly between relevant and irrelevant variables within groups.

Consider as an illustrative example the case with evenly sized non-overlapping groups, where of each relevant group of covariates $G \in \mathcal{G}_{\cap S}$, a subset of size $\Omega(\abs{G}^{\gamma_1})$, $\gamma_1 \in (0, 1]$ have the same fixed non-zero coefficient value, with the rest being 0. Here, $\abs{\tS} = \abs{S}\abs{G}^{1- \gamma_1} $. Suppose $\sigma^2\abs{S}\log (p) /n = o(n^{-\gamma_2})$ for some $\gamma_2 > 0$.

In this case, assuming appropriate conditions are met, the Adaptive Lasso estimate $\hbeta^{(a)}$ then achieves
\[\norm{X(\beta - \hbeta^{(a)})}_n^2 = O\left(\sigma^2 \frac{\abs{S}}{n} \left(1+  \sqrt{\log(p)n^{-\gamma_2}}  \right)^2\right).\]
For the Co-adaptive Lasso, we note that $\max_G \abs{G}^{\gamma_1}/\max(\norm{\beta_G}^2, \min_{H \in \mathcal{G}_{\cap S}} \norm{\beta_H}^2) = O(1)$, so B2 is satisfied with $\gamma_1, \gamma_2$. If the additional assumptions of Corollary~\ref{coraiao} are met,
\[
 \norm{X \left( \beta - \hbetaB\right)}^2_n = O\left( \sigma^2\frac{\abs{S} \abs{G}^{1-\gamma_1}}{n}(1+ \sqrt{\log \abs{\mathcal{G}}\abs{G}^{-\gamma_1} n^{-\gamma_2}})^2 \right),
\]
If $\sqrt{\log(p)n^{-\gamma_2}}  \rightarrow 0$, then the Adaptive Lasso attains the optimal rate of $O\left(\sigma\sqrt{\frac{\abs{S}}{n}}\right)$, and the Co-adaptive Lasso cannot improve on this. Otherwise, the co-adaptivisation improves things so long as $\abs{G}^{1-\gamma_1}  = o(\log(p)n^{-\gamma_2})$ and $\gamma_1 > 1/2$. 

In particular, if the number of non-zero variables in each group $G$ is a fixed proportion of the full group size, the Co-adaptive Lasso can always attain an optimal rate of $\sigma^2S/n$ if the initial Lasso has $o(1)$ prediction error, something that is not the case with the Adaptive Lasso.
%

\subsection{Comparison to the Group Lasso}

From \citet{lounicigroup}, it can be inferred that the Group Lasso produces prediction error bounds of the form
\begin{align*}
 \norm{X(\beta - \hbeta)}_n^2 &= O\left( \frac{\sigma^2 }{n} \sum_{G \in \mathcal{G}_{\cap S}} \left( \abs{G}  + \sqrt{ \abs{G} \log(p/\abs{G})} + \log(p/\abs{G})\right)\right) \\
&= O\left( \frac{\sigma^2}{n} \left(\abs{\tS} + \abs{\mathcal{G}_{\cap S}}\sqrt{\abs{G}\log\abs{\mathcal{G}}} + \abs{\mathcal{G}_{\cap S}} \log\abs{\mathcal{G}} \right) \right),
\end{align*}
in the case of non-overlapping, evenly sized groups.

Suppose Corollary~\ref{coraiao}'s conditions are satisfied. Then the Co-adaptive Lasso attains

\begin{align*}
 \norm{X \left( \beta - \hbetaB\right)}^2_n &= O\left( \sigma^2\frac{\abs{\tS}}{n}(1+ \sqrt{\log \abs{\mathcal{G}}\abs{G}^{-\gamma_1} n^{-\gamma_2}})^2 \right) \\
&= O\left( \sigma^2\frac{\abs{\tS}}{n}(1+ 2\sqrt{\log \abs{\mathcal{G}}\abs{G}^{-\gamma_1} n^{-\gamma_2}} + \log \abs{\mathcal{G}}\abs{G}^{-\gamma_1} n^{-\gamma_2}) \right) \\
&= O\left( \frac{\sigma^2}{n} \left( \abs{\tS} + \abs{\mathcal{G}_{\cap S}} \sqrt{\abs{G}\log\abs{\mathcal{G}} \abs{G}^{1-\gamma_1} n^{-\gamma_2}}   + \abs{\mathcal{G}_{\cap S}}\log \abs{\mathcal{G}}\abs{G}^{1-\gamma_1} n^{-\gamma_2} \right)\right).
\end{align*}

If $\abs{G}^{1-\gamma_1}n^{-\gamma_2} = o(1)$, the Co-adaptive Lasso is superior to the Group Lasso. In particular, if B2 is satisfied with $\gamma_1 = 1$, $\gamma_2 >0$, then the Co-adaptive Lasso attains quickly the $O\left(\sigma^2\abs{\tS}/n\right)$ rate.

The Co-adaptive Lasso holds a further advantage if within group sparsity exists. Suppose that the conditions of Corollary~\ref{corpropor} are satisfied. Then by that corollary,
\begin{align*}
 \norm{ X \left(\beta - \hbetaB\right)}^2_n &= O \left(\frac{\sigma^2}{n} \max(\log\abs{\mathcal{G}} \abs{G}^{-\gamma_1}n^{-\gamma_2} , \log\abs{\tS} ) \right) \\
&\leq O\left( \frac{\sigma^2}{n} \left( \abs{S}\log\abs{\tS} + \frac{\abs{S}}{\abs{G}^{\gamma_1}} \log\abs{\mathcal{G}}n^{-\gamma_2} \right)\right).
\end{align*}

This can be a substantial improvement if $\abs{S}\log\abs{\tS}$ is much smaller than $\abs{\tS}$, and $\abs{S}\abs{G}^{-\gamma_1}n^{-\gamma_2}$ is much smaller than $\abs{\mathcal{G}_{\cap S}}$. In particular, if $\abs{\tS}/n$ fails to converge, the Co-adaptive Lasso can succeed if $\abs{S}$ diminishes quickly enough, while we cannot usually expect success with the Group Lasso.

Nevertheless, the Group Lasso can do better if the conditions for the Co-adaptive Lasso are too difficult to satisfy, especially if the initial Lasso fails completely to converge, or if the coefficients in each group are very small. In particular, the multi-task learning context analysed in \citet{lounicigroup}, where a set of separate regressions are related by a common sparsity pattern, fails to converge for the Lasso as the number of tasks alone increases, and so will generally fail with the Co-adaptive Lasso.

\section{Detailed properties}\label{sec:detail}

For our theoretical analysis, let us assume there exists $\beta = \beta_S$ such that the data can be written as
\[Y = X\beta + \varepsilon\]
where $S$ is a set of relevant covariates, and $\varepsilon$ is a noise term. In our analysis, we assume that the model is ``truly'' sparse in the sense that the underlying model $\beta$ has non-zeroes only in $S$. It is possible to encompass the more general case where $\beta$ is only approximately sparse by using $\varepsilon$ to incorporate approximation error arising from sparsification, if the removed covariates have a very small coefficient.

Our results will be in two halves - first, we show that under a broad condition, the weights $w$ successfully separate between covariates from zero groups and covariates from non-zero groups. Then, we show that with this separation, we achieve good results on the second optimisation.

\subsection{Convergence of Group Weights}\label{sec:weightscon}

Existing work \citep{Bickel} \citep{compatibility} have proven bounds for the estimation error of the Lasso. We show as a variant bounds on the group-wise errors of the initial estimate, and so by implication, the weights used in the second estimate.

\begin{lemma}\label{groupconverge}
Let
\[\lambda \geq 2 \max \abs{X'\varepsilon/n}.\] 

Then for all $G \in \mathcal{G}$,
\[
 \norm{ \hbetaA_{G} - \beta_{G}}/\sqrt{\abs{G}} \leq \frac{2 (\lambda +\max\abs{X'\varepsilon/n})  \sqrt{\abs{S}} }{\phi_{\mathcal{G}}(3, S)\phi(3, S)\sqrt{\abs{G}}}.
\]
\end{lemma}

From Lemma~\ref{groupconverge}, the contribution of the noise $\varepsilon$ then is based on the maximal correlation $\max\abs{ X'\varepsilon/n}$. As other authors have identified \citep{compatibility}, it is possible to bound this for normally distributed noise:

\begin{lemma}\label{residualcorr}

Let $X \in \R^{n \times k}$, with $\norm{X^{(j)}}_n$ bounded above by some constant $C$ for all $j$, and $\varepsilon \sim N(0, \sigma^2)$. Then with probability exceeding 
\[1 - \frac{\exp(-t/2)}{\sqrt{\pi}(t +  2\log(k))},\]
we have that
\[
\max\abs{X'\varepsilon/n} \leq C\sigma\sqrt{\frac{t + 2\log(k)}{n}}.
\]
\end{lemma}

A similar lemma will work in the case of more general subgaussian noise, albeit with different constants in the bound.

Using the above, we have the following:

\begin{lemma}\label{weightsstandard}
 Suppose that A1, A2 and B1 hold and $\mathcal{G}$ is non-overlapping. Choose $\lambda = O\left( \max\abs{X'\varepsilon/n} \right)$ with $\lambda \geq 2 \max\abs{X'\varepsilon/n}$. Then, for all $\eta \geq 0$, there exists $C'$ such that with probability exceeding $1-\eta$ we have for each $G \in \mathcal{G}$,
\begin{eqnarray}
\left\vert w^{-1}_{G} - \frac{ \norm{\beta_{G}}}{\sqrt{\abs{G}}} \right\vert &\leq C'\left(\sigma\sqrt{ \frac{\abs{S}\log(p) }{n\abs{G}}}\right)
\end{eqnarray}
where the inequalities are taken term-wise.
\end{lemma}
\begin{proof}
 Result follows trivially from combining Lemma~\ref{groupconverge} and Lemma~\ref{residualcorr}, noting that $w_j^{-1} = \norm{\hbetaA_G}/\abs{G}$ for $j \in G$.
\end{proof}

\subsection{Second stage convergence}

To translate the bounds on the weights into bounds on the second stage, we require some theorems for the weighted Lasso. Now, several authors have proven a variety of results relating to this. In particular, \citet{adaptiveLASSOgeer} proved some inequalities similar in spirit to ours. However, their focus was on convergence for the Adaptive Lasso, with the additional complication of model misspecification. The Co-adaptive Lasso provides a distinct challenge in that we expect faster weight convergence for the variables in the out of group set $\tS^c$, and slow or non-existent weight convergence in the case of irrelevant variables within groups.

\begin{lemma}Results for the second stage \label{secondstage}

Fix any $T$, $\tS \supseteq T \supseteq S$ such that $\hbetaOT \in \R^p$, the ordinary least squares estimate when computed on the restricted covariate set $T$, exists.
\[ 
 \hbetaOT_{T^c} = 0, \quad \hbetaOT_T = (X_T'X_T)^{-1}X_T'Y,
\]
with residual $\varepsilon^{OLS} = Y - X\hbetaOT$.

Let $\mu^{\varepsilon_{OLS}} = \max \abs{ X'\varepsilon^{OLS}}/n$, and $\mu^{\varepsilon_{OLS}}_{\tS} = \max \abs{ X_{\tS \setminus T}'\varepsilon^{OLS}}/n$.

For $\mu \geq \max \left( \mu^{\varepsilon_{OLS}}/ w^-_{\sco},\mu^{\varepsilon_{OLS}}_{\tS}/w^-_{\scit} \right)$, setting

\begin{align*}L_1 &=  \mu  w^+_T / \left(\mu w^-_{\sco} - \mu^{\varepsilon_{OLS}} \right) \\
L_2 &=  \mu  w^+_T / \left( \mu  w^-_{\scit} -\mu^{\varepsilon_{OLS}}_{\tS} \right), 
\end{align*}
we have that
\begin{align}
 \norm{ X(\hbetaOT  - \hbetaB)  }_n &\leq  2\mu  w^+_T \sqrt{\vert T \vert}/ \max \left(\phi(L_1, \tS), \phi(\max(L_1,L_2), T)\right) \\
\norm{ \hbetaOT_T  - \hbetaB_T  } &\leq 2 \mu  w^+_T \sqrt{\vert T \vert}/ \max \left(\phi^2(L_1, \tS), \phi^2(\max(L_1,L_2), T)\right) \\
 \norm{\hbetaOT  - \hbetaB} &\leq 2\mu w^+_{T} \sqrt{\vert T \vert}\frac{1 +\max(L_1, L_2) }{  \max \left(\phi^2(L_1, \tS, \abs{\tS} + \abs{T}), \phi^2(\max(L_1, L_2), T,2\abs{T} )\right)} . \label{OLSconverg}
\end{align}

 \end{lemma}

\begin{remark}\label{remarkharm}
Note that in all of these theorems, a somewhat better bound and conditions can be obtained by using the Cauchy-Strauss bound to replace $ w^+_T $ with $\norm{w_T/\sqrt{T}}$ and $ w^+_{\tS}$ with $\norm{w_{\tS}/\sqrt{\tS}}$ in both the inequalities and the calculation of $L_1$ and $L_2$. This can improve things in the case where a small number of groups in the signal $\beta$ are significantly smaller in terms of group-wise L2 norm than the rest, and the number of groups is large. However, in this paper, we use the former bound for simplicity, as the latter bound leads to conditions on sparsity-weighted harmonic means of group-wise L2 norms that are more difficult to interpret.
\end{remark}

The effect of Lemma~\ref{secondstage} can be examined by varying $T$. Setting $T$ to equal $\tS$, we see that as the weights on the out of group variables increase relative to the ones on the relevant groups, we can obtain a fast convergence to the ordinary least squares estimate restricted to variables in the same group as the true ones, assuming that this exists. We do this by selecting a tuning parameter that scales in a manner inversely proportional to the weights on the out of group variables. In other words, by paying the price of ultimately just doing a least squares regression on $X_{\tS}$, we remove the influence of the remaining variables very easily. This case is especially useful in the AIAO case, or if the group sizes are quite small.

Meanwhile, setting $T$ to equal $S$ implies that by choosing a higher tuning parameter than before, we can attempt to take advantage of the within group sparsity as well. At best, we can obtain a result similar to conducting a Lasso restricted to the relevant groups $\tS$, a substantial improvement if $\abs{\tS} \ll p$. We pay a price in this case in terms of the ratio $w_S^+/w_{\sci}^-$, which can be quite significant if impact of the relevant groups varies greatly.

The conditions required for these convergences offer two possibilities. One is to bound $\phi(L_1, \tS)$ away from zero, which becomes increasingly easy as the weights ratio between relevant and irrelevant groups increase, since $L_1$ decreases. However, this requires that $\tS$ be not too large, and especially be smaller than $n$, or else the loss of identifiability means the condition is automatically failed. The second condition of bounding $\phi(\max(L_1, L_2),T)$ allows reasonable success in the case of large $\tS$, potentially larger than $n$, something distinct from the Group Lasso. Its fulfilment is more complex -- if the weights are the same within each group, $L_2 \geq 1$. Indeed, $L_2$ can be quite large if the contribution of relevant groups are quite variable, making this condition harder to satisfy and so Lemma~\ref{secondstage} harder to apply. We can compensate however by increasing $T$ to incorporate elements of $\sci$ that are likely to have small values of $w$. In other words, we can still have good results, if we are willing to accept that we will likely incorrectly select some irrelevant covariates in groups that appear collectively very relevant from the initial calculation.

Our main result then emerges by combining this theorem with the convergence results from Section~\ref{sec:weightscon}.

\section{Overlapping groups}\label{sec:overlap}

In general, groups in the case of group sparsity cannot be assumed to be non-overlapping. However, the existence of overlaps amongst groups poses a problem to the practitioner as to how to handle these overlaps. It is then necessary to tailor the algorithm to deal with overlapping groups in the appropriate fashion.

For the Group Lasso penalty $P_{group}(\beta) = \sum_G \norm{\beta_G}$:
\[
 \frac{\partial}{\partial \beta_i} P_{group}(\beta) = \sum_{G \ni i} \frac{ \beta_i}{\norm{\beta_G}}
\]

has a singularity at $\beta_i = 0$ if and only if there exists a group $G$ with $i \in G$ and $\norm{\beta_G} = 0$. As the presence of singularities indicate 'corners` in the penalty for which exactly zero estimates are possible, we have that allowable sparsity patterns take the form of intersections of complements of overlapping groups, with coefficients in overlaps between groups especially unlikely.

While this is useful in some cases \citep{structuredsparsity}, alternative interpretations of group overlaps might be more desirable in others. In the general case, a typical application is to find signals that are unions of groups.  \citet{overlapgroup} proposes one way to accommodate this situation. However, \citet{overlaptheo} established that there are several problems with this formulation - the computational cost can be very burdensome, and the conditions required for success are stringent, with generally poor results if the structure of the groups are too complex. Particular examples raised were nested group structures, and presence of sparsity in true groups.

However, in the co-adaptive framework, a flexible, and natural alternative framework can be constructed to deal with overlapping groups. For example, we can replicate a Group Lasso style behaviour for the Co-adaptive Lasso by calculating weights as, for each $j = 1, \ldots, p$, 

\[
 w_j = \sum_{G \ni j} \sqrt{\abs{G}}/\norm{\beta_G}.
\]

For behaviour involving selecting unions of groups, a variety of methods for choosing weights are possible, and we suggest here two possibilities for various scenarios:
\subsection{Overlapping group norm minimisation}

One approach to improving the performance of the overlapping Group Lasso, suggested in \citet{overlaptheo}, is to calculate the adaptive overlapping Group Lasso. Specifically, \citet{overlaptheo} suggests using the OLS estimate $\hbeta^{OLS}$ to calculate group weights, with, for some $\gamma < 0$, $w_G = \norm{v_G}^{\gamma}$ with
\[
\{v_G\} = \arg \min \sum_G \norm{v_G} \quad s.t. \mathrm{supp}(v_G) \subset G \forall G,  \sum {v_G} = \hbeta^{OLS}.
\]

Under this, and assuming some conditions -- in particular, that the above decomposition is unique in a neighbourhood around around the true value, and the overlap norm minimising decomposition of the true value itself is tight around the true relevant covariates -- they are able to prove consistency in the fixed design case where $n$ alone goes to infinity.

Such a strategy may be similarly applied to the Co-adaptive Lasso. For example, by conducting a weighted Lasso using for each covariate the minimal weight amongst groups containing the covariate, the same proof given in \citet{overlaptheo} will suffice to show the same consistency result under the same conditions. A second modification would be to use $\hbetaB$ instead, which would allow use when $p \geq n$. 

On the other hand, group norm minimisation weights has a variety of shortcomings. As discussed in \citet{overlaptheo}, the uniqueness criterion is strong, and rules out possibilities like nested groups. One potential improvement would be to apply the overlapping group norm instead as a penalty, treating the weight calculation as a Group Lasso problem itself, using $\hbetaB$ to replace $Y$, and $X$ as an identity matrix with repeated columns for variables appearing in more than one group, and the same grouping structure. Here the minimisation result corresponds to the case where the tuning parameter is taken to zero. 

Assume that the decomposition of the true $\beta$ that minimises the overlapping norm gives weights $w_{0,G}$. Using an application of Proposition 1 in \citet{overlaptheo}, and an appropriately chosen tuning parameter value, we can then obtain a bound

\[
 \sum_{G} \abs{ {w_G}^{-\gamma} -{w_{0,G}}^{-\gamma}}  \leq 32 \max_{G} \norm{\hbetaA_G - \beta_G}/ \kappa^2,
\]

where $1/\kappa^2 \in \left[\abs{\mathcal{G}_S}/16 ,4\abs{\mathcal{G}}\right]$ is a constant depending only on the complexity of $\mathcal{G}$. While this approach can be effective, choosing the correct tuning parameter may be difficult, and may not be useful for many types of grouping structure. 


\subsection{Maximum covariate-wise group norm}

A simple alternative approach is to take the initial estimate, calculate the grouped L2 norms $\norm{\beta_G}$ for each group $G$, and weigh each covariate $j$ as 
\[ 
w_j = \min_{G \ni j} \sqrt{\abs{G}}/\norm{\beta_G}.
\]

This approach reduces to the previous case when the groups do not overlap. Lemma~\ref{groupconverge} applies equally to this weight formulation, and we have then that by an analogous proof to Lemma~\ref{weightsstandard}, the following:

\begin{lemma}\label{weightoverlap}
Suppose that A1, A2 and B1 hold. Choose $\lambda = O(\max\abs{X'\varepsilon/n})$ with $\lambda \geq 2 \max\abs{X'\varepsilon/n}$. Then for all $\eta \geq 0$, there exists $C'$ such that with probability exceeding $1-\eta$ we have for each $j \in \{1,\ldots, p\}$, 
\[
  \left\vert  w_j^{-1} - \frac{\norm{\beta_{G_j}}}{\sqrt{\abs{G_j}}}\right\vert \leq C'\left(\sigma \sqrt{ \frac{\abs{S} \log(p)}{n \abs{G_j}}}\right), \quad G_j = \arg \max_{G \ni j}  \norm{\beta_G}/\sqrt{\abs{G}}.
\]
\end{lemma}

We can then adapt Theorem~\ref{bigmain} to use this lemma instead of Lemma~\ref{weightsstandard}, deriving similar bounds.

The problem with this approach is that the lightly penalised covariates -- that is, the variables $j$ for which $w_j \not\rightarrow \infty$ as $\hbetaA_j \rightarrow \beta$ will correspond to $\cup \mathcal{G}_{\cap S}$, which can be potentially very large, especially if there exists large groups that overlap significantly with $S$. In that case, the co-adaptive re-weighting will eliminate an insufficient proportion of the variables, thus leading to not a very good fit. Note that this is a problem similarly present with the overlapping Group Lasso.

One possible way to fix this issue is if we know or can infer the degree of overlap of groups with $S$, or whether there exists a covering set of groups for which there is very little within group sparsity. Then, we can compute instead of $\norm{\beta_{G_j}}^2$, appropriately trimmed or winzorised sums of $\beta^2_{G_j}$. By this method, we reduce the mistakenly picked up overlapping groups to only those with large overlaps with the truly relevant ones, paying the price of potentially losing covariate groups that have very high levels of within group sparsity.

\section{Empirical results}\label{sec:simulations}

We shall investigate the performance of the Co-adaptive Lasso in a variety of datasets, both real and generated.

\subsection{Splice site prediction}

As considered in \citet{meiergroup}, the splice problem concerns the prediction of splice sites -- the regions between coding (exons) and non-coding  (introns) DNA segments. In particular, the task of predicting `donor' splice sites -- the 5' splice site end of introns -- has been commonly used to demonstrate the Group Lasso \citep{meiergroup}, \citep{rothgroup}.

As in \citet{meiergroup}, we consider the MEMset dataset \citep{spliceyeoburge}, which consists of sub-sequences of DNA that contain the consensus position ``GT'', which are either true donor splice sites, or otherwise. Removing the consensus position then gives sequences of length 7 with 4 levels ($\{ A, C, G, T \}$). The original dataset consists of 8415 true and 170438 false human donor sites, with an additional test set of 4208 true and 89717 false donor sites. As per the original, we use the training set to build a smaller balanced training dataset with 5610 true and 5610 false donor sites, and an unbalanced validation set with the remainder, having the same ratio of true to false sites as the test set. This is done through sampling at random, without replacement.

Our goal is accurately predict whether a candidate site from the test set is a true or a false donor splice set. This is measured by the maximum attainable correlation coefficient \citep{spliceyeoburge} between a predicted vector and the true one,
\[
 \mathrm{cor}_{\max} = \max_{\tau} \mathrm{cor} \left( y_{test} , I\{ \widetilde{p}(x_{test}) \geq \tau \} \right).
\]

We consider dummy variables as in \citet{meiergroup}, using scaled versions of the indicator variables of ($\{ A, C, G \}$), plus all interactions up to three way. We train on the balanced dataset, using the validation set to choose tuning parameters -- unlike \citet{meiergroup}, we select according to maximising $\mathrm{cor}_{\max}$ on the validation set instead of maximising likelihood. This allows us to dispense with the need to correct the intercept.

We compare logistic versions of the Lasso, the Group Lasso, the Adaptive Lasso, and the Adaptive Group Lasso, and the Co-adaptive Lasso. For the two stage procedures, we use the same training and validation set for each stage of the procedure.

\begin{figure}[ht]
\vspace{-10pt}
 \centering
 \includegraphics[scale = 0.6]{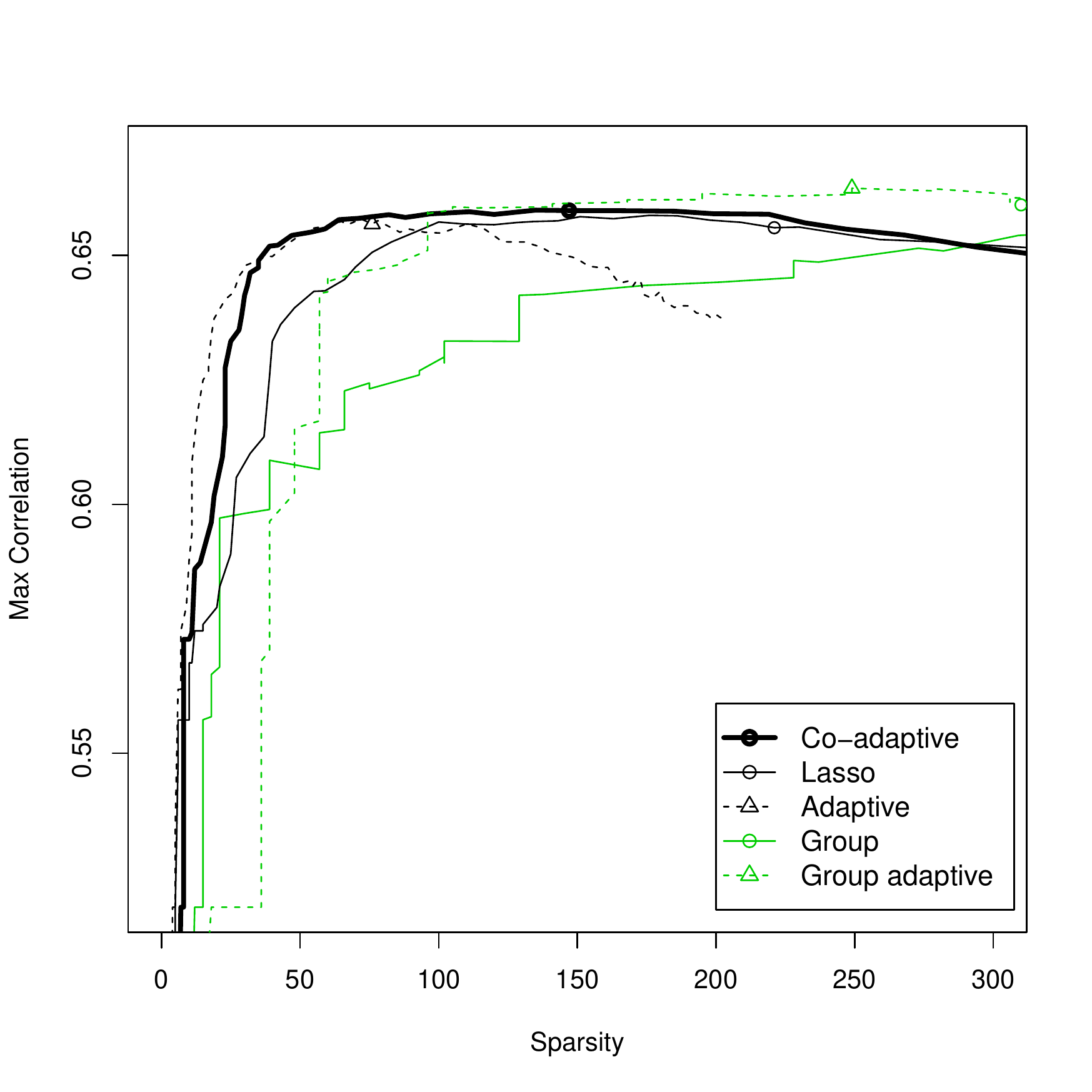}
\vspace{-10pt}
 \caption{Sparsity level vs $\mathrm{cor}_{\max}$ on the test set. The points on each curve represents the model chosen by reference to the validation set. The chosen model for the Group Lasso uses 966 variables and so is omitted from the graph for presentation reasons.}
 \label{fig:splicespar}
\end{figure}
\begin{figure}[ht]
\vspace{-10pt}
 \centering
 \includegraphics[scale = 0.6]{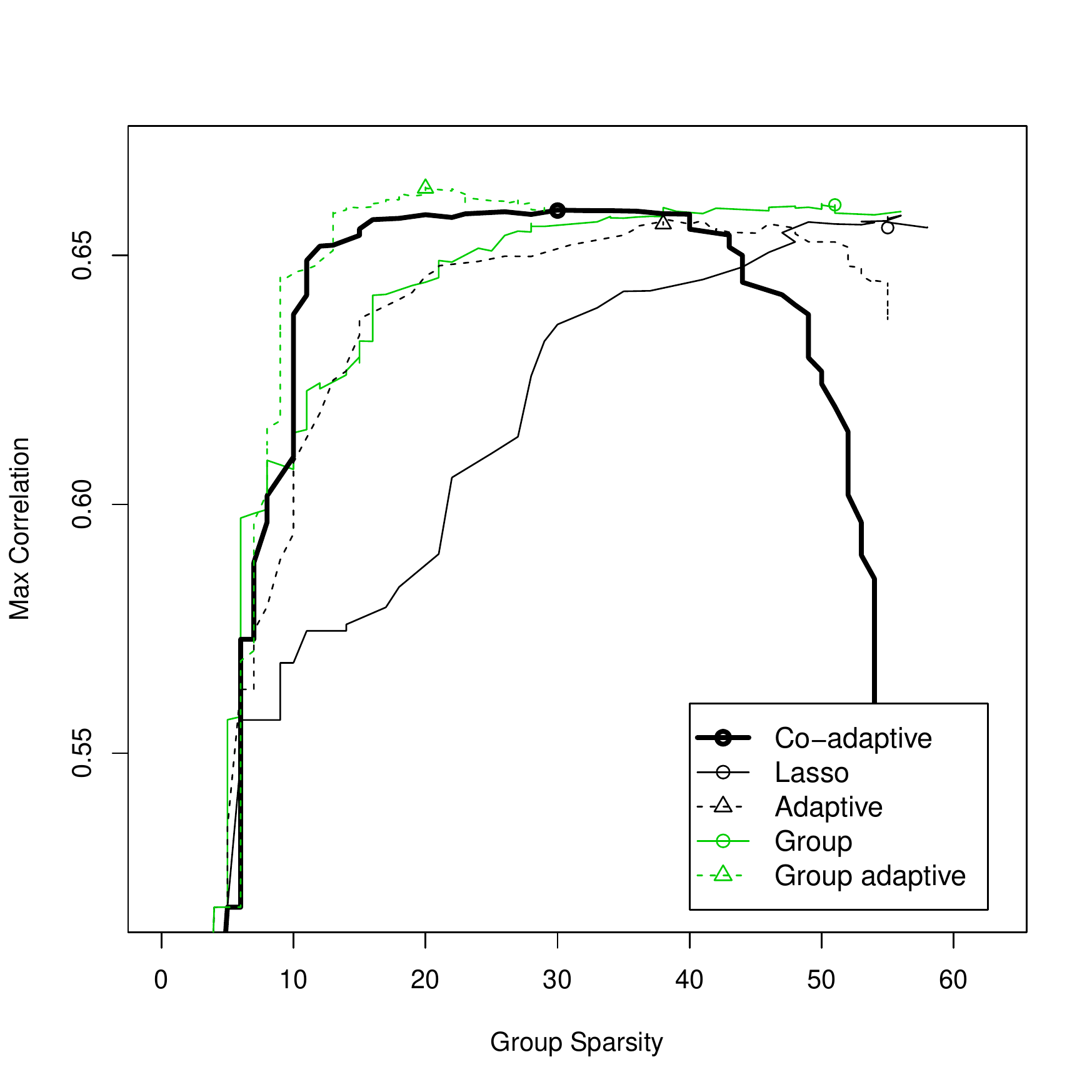}
\vspace{-10pt}
 \caption{Group sparsity level vs $\mathrm{cor}_{\max}$ on the test set. The points on each curve represents the model chosen by reference to the validation set.}
 \label{fig:splicegrpspar}
\end{figure}

\begin{center}
\begin{tabular}{l|l}
Algorithm&  $\mathrm{cor}_{\max}$ \\
\hline
 Co-adaptive Lasso& 0.659\\
 Lasso& 0.656\\
 Adaptive Lasso& 0.656\\
 Group Lasso& 0.660\\
 Adaptive Group Lasso& 0.664
\end{tabular}
\end{center}

Overall, the attained $\mathrm{cor}_{\max}$ is highly similar for all algorithms. The Adaptive Group Lasso achieves the best results, with the Group Lasso and Co-adaptive Lasso having similar results. The Adaptive Lasso and Lasso perform less well, but still reasonably. These results suggest there is not much within group sparsity on this dataset - the high level of correlation (which, on the training set, reached a maximum of 0.88) mean there is some advantage in selecting redundant variable sets to reduce variability on the test set.

More interesting is the breakdown of performance according to group sparsity and sparsity. Figure~\ref{fig:splicespar} and Figure~\ref{fig:splicegrpspar} give these results. We see that while the Group Lasso and Adaptive Group Lasso perform well overall, they require a much large set of selected variables to do so. The Lasso and the Adaptive Lasso meanwhile select a small number of variables, but fail in terms of attaining good results at any level of group sparsity. The Co-adaptive Lasso therefore managed to attain an excellent trade-off, by attaining much better within group sparsity, while still remaining competitive with the Adaptive Lasso in achieving good group sparse solutions. Indeed, for the most group sparse solutions, it even beats the Group Lasso.

Similar results were obtained when indicator variables including `T` were used to create an intentionally over-specified design matrix. In our experiments, the Lasso based algorithms were the fastest, though the majority of time was consumed with handling the large dataset. It is likely that an online formulation would have been more efficient.

\subsection{Artificial datasets -- Non-overlapping}

We set up a series of simulation experiments to test the Co-adaptive Lasso against a range of similar algorithms. In particular, we compare against

\begin{itemize}
 \item The Lasso \citep{tibLASSO}
\item The Adaptive Lasso \citep{adaptiveLASSO}
\item The Group Lasso \citep{meiergroup}
\item The Sparse Group Lasso \citep{friedmanSGL}.
\end{itemize}

We generate a range of scenarios showing group sparsity and varying levels of within group sparsity in the linear regression context. Specifically, we generate
\[
Y = X \beta + \varepsilon,
\]

with X as a $n \times p$ matrix from an i.i.d. standard normal distribution. Then, choosing non-overlapping groups with group size $\abs{G}$, we choose $S$ as random subsets of 2 selected groups, with the same within group sparsity level in each group. We repeat each scenario with the final selected variables $\beta_S$ as either independently standard normal or constant at 1. Finally, we generate the noise $\varepsilon$ as independent normal, with variance chosen so that each iteration would have a SNR of 2.

Our scenarios are:
\begin{enumerate}
\item $n = 150$, $p = 2000$, $\abs{G} = 10$, $\abs{S} = 10$ 
\item $n = 150$, $p = 2000$, $\abs{G} = 10$, $\abs{S} = 20$ 
\item $n = 150$, $p = 2000$, $\abs{G} = 100$, $\abs{S} = 10$
\item $n = 500$, $p = 2000$, $\abs{G} = 10$, $\abs{S} = 20$ 
\item $n = 500$, $p = 2000$, $\abs{G} = 100$, $\abs{S} = 10$. 
\end{enumerate}

Hence, Scenario 2 and 4 are AIAO situations, while the remaining scenarios have some degree of within group sparsity. 

We conduct 100 iterations of each scenario, and compute estimates choosing tuning parameters by 10 fold cross validation. In the case of the SGL, which has 2 parameters, we use the default implementation which fixes the mixing parameter $\alpha$ at 0.95. We compare the estimation error $\norm{\hbeta - \beta}^2$ -- given the independence of our covariates, this is equivalent to comparing the expected error on a new test set, minus the contribution from the new $\varepsilon$. Table~\ref{tabsims} shows the results of our simulations.

\begin{table}
{%
\newcommand{\mc}[3]{\multicolumn{#1}{#2}{#3}}
\begin{center}
\begin{tabular}{l|l|l|l|l|l|l|l|l|l|l}
 & \mc{2}{l|}{Scenario 1} & \mc{2}{l|}{Scenario 2} & \mc{2}{l|}{Scenario 3} & \mc{2}{l|}{Scenario 4} & \mc{2}{l}{Scenario 5}\\
\hline
 & Const & Norm & Const & Norm & Const & Norm & Const & Norm & Const & Norm\\\hline
Lasso &   2.81 & 2.11 &16.2 & 7.64 & 2.84 & 1.88 & 2.10 & 1.74 & 0.57 & 0.50 \\\hline
Adaptive& 1.35 & 1.14 &15.87 & 5.48 & 1.38 & \textbf{1.11} & 1.08 & 1.01 & \textbf{0.21} & \textbf{0.20} \\\hline
Group Lasso & 1.09 & 1.09 & \textbf{2.25} & 2.09 & 8.45 & 7.50 & 0.62 & 0.58 & 2.12 & 1.93 \\\hline
SGL  &1.84 & 1.39 &10.05 & 5.09 & 2.09 & 1.42 & 1.47 & 1.20 & 0.41 & 0.34  \\\hline 
Co-adaptive & \textbf{0.55} & \textbf{0.53} & 3.51 & \textbf{1.21} & \textbf{1.35} & 1.40 & \textbf{0.37} & \textbf{0.34} & 0.36 & 0.39
\end{tabular}
\end{center}
}%
\caption{Results for simulated datasets. The table shows estimation error. The best performer in each scenario is highlighted in bold.} \label{tabsims}
\end{table} 

From these results, the Co-adaptive Lasso performs well in all of the scenarios. In most, it performs the best, or nearly the best, with the exceptions of scenario 5 and scenario 2, constant version.  In the former case, understandably it performs less well than the Adaptive Lasso because the groups in this case are not very informative, while the sample size is large so the initial Lasso provides good weights for the adaptive second stage. In the latter case, there is no within group sparsity, and the inherent L2 penalty of the Group Lasso helps encourage it to make an estimate that is more constant in magnitude within groups, which matches the true signal. In comparison, the SGL, which is the other algorithm attempting within group sparsity, performs roughly in between the Group Lasso and the Lasso. It needs to be noted that this behaviour was observed using the default values for the second tuning parameter - more success might be reached by selecting this, for instance, through a grid based cross validation. However such an approach would necessarily greatly increase the time require for computation.

In terms of computational time, the Lasso based methods were much faster than the Group Lasso or SGL. However, this may be due to the specific implementation of the algorithms.

\subsection{Artificial datasets -- Overlapping}

We conduct a second set of simulations with overlapping groups. Here, we compare the Co-adaptive Lasso with the minimum group-wise norm weights to 

\begin{itemize}
 \item The Lasso \citep{tibLASSO}
\item The Adaptive Lasso \citep{adaptiveLASSO}
\item The Overlapping Group Lasso \citep{overlapgroup}.
\end{itemize}

In this case, we generate again
\[
 Y = X \beta + \varepsilon,
\]
with $n = 500$, $p = 2000$. We choose $S = \{1, \ldots, 20\}$, and choose as two scenarios $\beta_S$ either standard normal or constant at 1. We generate $\varepsilon$ as before to be independent normal with variance chosen to give SNR of 2.

We give the data a more complex group structure: $\mathcal{G} = \{\seq{G_1}{G_{100}}, \seq{G_{101}}{G_{120}} \}$, with
\begin{align*}
 G_i &= \{  \seq{20(i-1) +1}{20(i-1) +20}\} , &\quad   \textrm{ for } i= 1, \ldots, 100 \\
 G_i &= \{  1, \ldots, 100 \} \times (i-100) , &\quad \textrm{ for } i= 101, \ldots, 120 .
\end{align*}

In this setup, each of $\seq{G_1}{G_{100}}$ overlaps with each of $\seq{G_{101}}{G_{120}}$, though only $G_1$ is required to cover the signal. We conduct again 100 iterations of these scenarios, using 10 fold cross validation to select tuning parameters and comparing the estimation error. Table~\ref{tabsimsover} gives the results.

\begin{table}
{%
\newcommand{\mc}[3]{\multicolumn{#1}{#2}{#3}}
\begin{center}
\begin{tabular}{l|l|l}
 & Const & Norm \\\hline
Lasso & 4.3 &3.5 \\\hline  Adaptive& 2.5 &2.3  \\\hline Group Lasso & 1.4 &1.4  \\\hline Co-adaptive & 1.1 &1.2
\end{tabular}
\end{center}
}%
\caption{Results for simulated datasets. The table shows estimation error. The best performer in each scenario is highlighted in bold.} \label{tabsimsover}
\end{table} 

Once again, the Co-adaptive Lasso is superior in both cases. We note that this happens despite the fact that $\tS$ in this case is the entire set $\{1, \ldots p\}$. This is likely because in this scenario, even though the groups $\seq{G_{101}}{G_{120}}$ overlap with $S$, the average signal on them $\norm{\beta_G}/\sqrt{\abs{G}}$ is much smaller than for $G_1$.

\section{Discussion}

We have defined a fast calculating method of conducting variable selection under a group structure, that facilitates the use of within group sparsity. The Co-adaptive Lasso can be easily coded using any existing methods of calculating the standard Lasso, including online computation methods. We have proven some convergence properties for the method, and illustrated its competitiveness relative to several state of the art methods. The procedure may be applied to a range of contexts.

Several areas warrant further investigation. The Co-adaptive Lasso with its framework of re-weighting Lassos can be fairly easily extended to include more complex re-weighting schemes. These may allow the utilisation of more complex subject specific prior information. For instance, in All-In-All-Out cases where the within group sparsity is low and the signal is spread out relatively evenly amongst members of each group, robust statistics related methods of weight computation have worked well to screen out single-variable mistakes made by the initial Lasso calculation. We have suggested this in the case of overlapping groups, though this can be useful more generally.

Several previous examinations of concepts similar to the Co-adaptive Lasso have focused on repeating the procedure until convergence. The general idea of those procedures is to attain local convergence to the optimum of some implied concave penalised maximum likelihood problem, but because these penalised maximum likelihood problems have characteristically multiple optima, it's not clear whether this would imply good properties, despite the necessary increase in computational cost. Moreover, more complex re-weighting procedures fall out of this paradigm, because there can sometimes be no compatible implied penalty function. Further investigation may go into when these iterative procedures can improve performance.

\appendix
\section{Appendix: Proofs of theorems}

\subsection{Proof of Lemma~\ref{REineq}}

\begin{proof}
Let $\delta$ be any vector satisfying 
\[
 \frac{\norm{X \delta}^2_n}{\norm{\delta_S}^2} = \phi^2(L, S), \quad \norml{\delta_{S^c}} \leq L \sqrt{\abs{ S }} \norm{\delta_S}.
\]

Assume without loss of generality that $\norm{X\delta}_n = 1$. Then by Definition~\ref{groupre}, for each $H \in \mathcal{H}$, $\norm{\delta_H}^2 \leq 1/\phi_G^2(L, S).$

Hence 
\[
 1/\phi^2(L,S) = \norm{\delta_S}^2 \leq \sum_{H \in \mathcal{H}} \norm{\delta_H}^2\leq \abs{\mathcal{H}}/\phi_G^2(L, S).
\]

Similarly, if $\delta$ satisfies
\[
 \min_{M\in \mathcal{G}} \frac{\norm{X \delta}^2_n}{\norm{\delta_M}^2} =\phi_\mathcal{G}^2(L, S), \quad \norml{\delta_{S^c}} \leq L \sqrt{\abs{ S }} \norm{\delta_S},
\]
assuming $\norm{X\delta}_n = 1$ means by definition that
\[
 1/\phi_\mathcal{G}^2(L, S) \leq \max_{M \in \mathcal{G}} \norm{\delta_M}^2 \leq \max_{M \in \mathcal{G}} \norm{\delta_{M \cup S}}^2 \leq 1/\phi^2(L, S, \abs{S} + \max \abs{G}).
\]

Finally, if $\delta$ satisfies 
 \[
  \min_{M } \left\{ \frac{\norm{X \delta}^2_n}{\norm{\delta_M}^2} :  S \subset M, \abs{M} \leq \abs{S} +\max\abs{G}\right\} =\phi^2(L, S, \abs{S} + \max\abs{G} ), \quad \norml{\delta_{S^c}} \leq L \sqrt{\abs{ S }} \norm{\delta_S},
 \]
assuming $\norm{X\delta}_n = 1$ means for some $M$, $\abs{M \setminus S} = \max\abs{G}$, $\norm{\delta_M}^2 = 1/ \phi^2(L, S, \abs{S} + \max\abs{G} )$.

Then
\begin{align*}
\norm{\delta_S}^2 &= \norm{\delta_M}^2 - \norm{\delta_{M\setminus S}}^2 \\
&\geq \norm{\delta_M}^2  - \norm{\delta_{M\setminus S}}_1^2 \\
&\geq \norm{\delta_M}^2  - L^2 \abs{S}\norm{\delta_S}^2.
\end{align*}
So
\[
 \frac{1}{\phi^2(L,S)} \geq \norm{\delta_S}^2 \geq \frac{\norm{\delta_M}^2}{1 + L^2 \abs{S}} \geq \frac{1}{\phi^2(L, S, \abs{S} + \max\abs{G}) (1 + L^2 \abs{S})}.
\]
\end{proof}

\subsection{Proof of Lemma~\ref{groupconverge}}

\begin{proof}
%
%
%
%
Similarly to Lemma~\ref{secondstage}, we have that
\begin{align*}
\lambda\norml{ \beta} - \lambda\norml{\hbetaA} &\geq \frac{1}{2}\left(\norm{Y - X\hbetaA}_n^2 - \norm{Y - X\beta}_n^2\right) \\
&\geq \frac{1}{2} \norm{ X(\hbetaA - \beta) }_n^2 - \left(\varepsilon, X\left(\hbetaA - \beta \right)\right)_n \\
&\geq \frac{1}{2} \norm{ X(\hbetaA - \beta) }_n^2 - \max\abs{X'\varepsilon/n} \norml{\hbetaA_{S^c}} -\max\abs{X'\varepsilon/n} \norml{\hbetaA_{S} - \beta_S} .
\end{align*}

Hence
\[
 \frac{1}{2} \norm{ X(\hbetaA - \beta) }_n^2 + \left(\lambda -  \max\abs{X'\varepsilon/n} \right) \norml{\hbetaA_{S^c}} \leq \left(\lambda + \max\abs{X'\varepsilon/n}\right) \norml{\hbetaA_{S} - \beta_S}.
\]

So by definition of $\phi(L,S)$, $\phi_{\mathcal{G}}(L,S)$, we have by a similar argument to Lemma~\ref{secondstage} that

\[
  \norm{ \hbetaA_G - \beta_G } \leq \frac{2(\lambda + \max\abs{X'\varepsilon/n}) \sqrt{\abs{S}} }{\phi(L, S)\phi_{\mathcal{G}}(L,S) }
\]
with

\[
 L = \frac{\lambda +  \max\abs{X'\varepsilon/n}}{\lambda - \max\abs{X'\varepsilon/n}} \leq  \frac{3 \max\abs{X'\varepsilon/n}}{\max \abs{X'\varepsilon/n}}
\]

As $\phi(L,S)$, $\phi_{\mathcal{G}}(L, S)$ decreases with $L$, the rest follows.
\end{proof}

\subsection{Proof of Lemma~\ref{residualcorr}}

\begin{proof}
Similar proofs have appeared elsewhere. We present the proof here for convenience.

Now, $X'\varepsilon/\norm{X}$ are distributed identically (but possibly not independently) Normal with variance $\sigma^2$. Therefore setting $\delta  = \sqrt{\frac{t + 2\log(k)}{n}}$, we have

\begin{align*}
P( \norm{X'\varepsilon}_\infty/n \geq \delta C\sigma) &\leq  \abs{T}P( \abs{X'\varepsilon}/(\norm{X}\sigma) \geq \sqrt{n}\delta)\\
&\leq \frac{\abs{T}}{ \sqrt{n\pi}\delta} \exp(-n\delta^2/2) \\
&\leq \frac{\exp( - t/2) }{ \sqrt{\pi(t + 2\log(k))}}.
\end{align*}
\end{proof}

\subsection{Proof of Lemma~\ref{secondstage}}

\begin{proof}

Our proofs are similar in spirit to those in \citet{adaptiveLASSOgeer}.

By definition of the weighted Lasso, we have that

 \begin{align*}
\mu\norml{w \hbetaOT} - \mu\norml{w\hbetaB} &\geq \frac{1}{2}\left(\norm{Y - X\hbetaB}_n^2 - \norm{Y - X\hbetaOT}_n^2\right) \\
&\geq \frac{1}{2} \norm{ X(\hbetaB - \hbetaOT) }_n^2 - \left(\varepsilon^{OLS}, X\left(\hbetaB - \hbetaOT \right)\right)_n \\
&\geq \frac{1}{2} \norm{ X(\hbetaB - \hbetaOT) }_n^2 - \mu^{\varepsilon_{OLS}}\norml{\hbetaB_{\sco}} - \mu^{\varepsilon_{OLS}}_{\tS}\norml{\hbetaB_{\scit}},
\end{align*}
since $X_{T}'\varepsilon^{OLS} = 0$.

Hence,
\begin{align}
 \frac{1}{2}\norm{ X(\hbetaOT_T  - \hbetaB)  }_n^2 &\leq  \mu\left(\norml{w_{T}\hbetaOT_T} - \norml{w_{T}\hbetaB_{T}}\right)  + \mu^{\varepsilon_{OLS}} \norml{\hbetaB_{\tS^c }} - \mu\norml{w_{\tS^c}\hbetaB_{\tS^c}}\nn\\
&\quad + \mu^{\varepsilon_{OLS}}_{\tS} \norml{\hbetaB_{\scit }} - \mu\norml{w_{\scit}\hbetaB_{\scit}}\nn\\
&\leq \mu  \norm{w_{T}} \norm{ \hbetaOT_T - \hbetaB_{T}} + \left(\mu^{\varepsilon_{OLS}} - \mu w^-_{\tS^c} \right)\norml{\hbetaB_{\tS^c}} \nn\\
&\quad+  \left(\mu_{\tS}^{\varepsilon_{OLS}} - \mu w^-_{\scit} \right)\norml{\hbetaB_{\scit}} \nn\\
&\leq \mu  w^+_{T} \sqrt{\abs{T}} \norm{ \hbetaOT_T - \hbetaB_{T}} + \left(\mu^{\varepsilon_{OLS}} - \mu w^-_{\tS^c} \right)\norml{\hbetaB_{\tS^c}} \nn\\
&\quad+  \left(\mu_{\tS}^{\varepsilon_{OLS}} - \mu w^-_{\scit} \right)\norml{\hbetaB_{\scit}} 
\end{align}

Let $\mu \geq \max \left( \mu^{\varepsilon_{OLS}}/ w^-_{\sco} ,\mu^{\varepsilon_{OLS}}_{\tS}/ w^-_{\scit}  \right)$.  Setting 

\begin{align*}L_1 &=  \mu   w^+_T / \left(\mu w^-_{\sco} - \mu^{\varepsilon_{OLS}} \right) \\
L_2 &=  \mu  w^+_T / \left( \mu w^-_{\scit} -\mu^{\varepsilon_{OLS}}_{\tS} \right), 
\end{align*}

we have then that

\begin{align*}
\sqrt{\abs{T}} \norm{ \hbetaOT_T - \hbetaB_{T}} &\geq \norml{\hbetaB_{\tS^c}}/L_1 + \norml{\hbetaB_{\scit}}/L_2 \\
\max ( L_1, L_2) \sqrt{\abs{T}} \norm{ \hbetaOT_T - \hbetaB_{T}} &\geq \norml{\hbetaB_{T^c}}
\end{align*}

and by a similar argument

\[ L_1\sqrt{\abs{\tS}} \norm{ \hbetaOT - \hbetaB_{\tS}} \geq \norml{\hbetaB_{\tS^c}}.\]

By definition of $\phi(\max(L_1,L_2), T)$ and  $\phi(L_1,\tS)$, observing in the latter case that $\norm{\hbetaOT_{\tS} - \hbetaB_{\tS}} \geq \norm{\hbetaOT_{T} - \hbetaB_{T}}$, we have that

\begin{align}
\frac{\mu   w^+_T  \sqrt{\vert T \vert} \norm{X(\hbetaOT - \hbetaB)}_n }{ \max \left(\phi(L_1, \tS), \phi(\max(L_1,L_2), T)\right)}   &\geq \frac{1}{2}\norm{ X(\hbetaOT  - \hbetaB)  }_n^2  + \left( \mu w^-_{\tS^c} - \mu^{\varepsilon_{OLS}}  \right)\norml{\hbetaB_{\tS^c}} \nn\\
&\quad+  \left( \mu w^-_{\scit}- \mu_{\tS}^{\varepsilon_{OLS}} \right)\norml{\hbetaB_{\scit}}. \nn
\end{align}
Hence
\begin{align}
 \norm{ X(\hbetaOT  - \hbetaB)  }_n &\leq  2\mu   w^+_T \sqrt{\vert T \vert}/ \max \left(\phi(L_1, \tS), \phi(\max(L_1,L_2), T)\right) \\
\norm{ \hbetaOT_T  - \hbetaB_T  } &\leq 2\mu   w^+_T \sqrt{\vert T \vert}/ \max \left(\phi^2(L_1, \tS), \phi^2(\max(L_1, L_2), T)\right) \\
\norm{ \hbetaOT_{T^c}  - \hbetaB_{T^c}  }_1 &\leq \max(L_1, L_2)2\mu   w^+_T \vert T \vert/ \max \left(\phi^2(L_1, \tS), \phi^2(\max(L_1, L_2), T)\right) \\
\norm{\hbetaOT -\hbeta} &\leq 2\left(1+\max(L_1, L_2) \sqrt {\abs{T}}\right)\mu   w^+_T  \sqrt{\vert T \vert}/ \max \left(\phi^2(L_1, \tS), \phi^2(\max(L_1, L_2), T)\right)
\end{align}

The third inequality can be quite poor for large $\abs{T}$. On the other hand, suppose that $\phi(L_1, \tS, \abs{\tS} + \abs{T})$ or $\phi(\max(L_1, L_2),T, 2\abs{T}) > 0$. Then for any $M \supset T$, $\abs{ M \setminus T} \leq \abs{T}$  ,  such that $\min \abs{\hbetaB_M } \geq \max \abs{\hbetaB_{M^c}}$, we have that

\[ \norm{\hbetaOT_M  - \hbetaB_{M}  } \leq 2 \mu  w^+_{T} \sqrt{\vert T \vert} /   \max \left(\phi^2(L_1, \tS, \abs{\tS} + \abs{T}), \phi^2(\max(L_1, L_2), T,2\abs{T} )\right) .\]

From Lemma 2.2 of \citet{compatibility}, however, we have that $\norm{\hbetaB_{M^c}} \leq \norml{\hbetaB_{T^c}}/\sqrt{\abs{ T}}$, so
\begin{equation*}
 \norm{\hbetaOT  - \hbetaB} \leq 2\mu w^+_{T} \sqrt{\vert T \vert}\frac{1 +\max(L_1, L_2) }{  \max \left(\phi^2(L_1, \tS, \abs{\tS} + \abs{T}), \phi^2(\max(L_1, L_2), T,2\abs{T} )\right)}  .
\end{equation*}

\end{proof}

\subsection{Proof of Theorem~\ref{bigmain}}
\begin{proof}
 Fix a $T \subseteq \tS$ for which the conditions are satisfied, and choose any $t > 0$ so that $3\exp(-t/2)/\sqrt{\pi} \leq \eta$  . Assume for now that $T \neq \tS$. Writing $P = I - X_{T}(X_{T}'X_{T})^{-1}X_{T}'$,
\[
\mu^{\varepsilon_{OLS}} = \max \abs{X'P\varepsilon}/n = \max \abs{(PX)'\varepsilon}/n, \quad \mu_{\tS}^{\varepsilon_{OLS}} = \max \abs{(PX_{\tS})'\varepsilon}/n, \quad \lambda^{\varepsilon} = \max \abs{X'\varepsilon}/n.
\]

But $P$ is a projection matrix, so for all $X^{(j)}$, $\norm{PX^{(j)}} \leq \norm{X^{(j)}}$. Therefore by Lemma~\ref{residualcorr}, under Assumption B1, each of the following fails to occur with probability less than $\exp(-t/2)/\sqrt{\pi}$: 
\begin{align*}
\mu^{\varepsilon_{OLS}} &\leq \sigma  \sqrt{\frac{t + 2 \log(p)}{n}} = O\left(\sigma \sqrt{\frac{\log(p)}{n}}\right) \\
\mu^{\varepsilon_{OLS}}_{\tS} &\leq \sigma  \sqrt{\frac{t + 2 \log\abs{\tS}}{n}} = O\left(\sigma \sqrt{\frac{\log\abs{\tS}}{n}}\right) 
\end{align*}


Further, by Lemma~\ref{weightsstandard} taking $\eta$ there to be less than $\exp(-t/2)/\sqrt{\pi}$, Conditions A1-2 and B1 implies then that choosing $\lambda = 2 \lambda^{\varepsilon}$, there exists some $C'>0$, such that with probability exceeding $1-\exp(-t/2)/\sqrt{\pi}$,
\begin{align*}
w_T^+ &\leq \max_{G \in \mathcal{G}_{\cap S}}  \left(  \norm{\beta_G}/\sqrt{\abs{G}} - C'\sigma \sqrt{\frac{ \abs{S}\log(p)}{n \abs{G}}} \right)^{-1} \\
w_{\tS \setminus T}^+ &\geq \min_{G \in \mathcal{G}_{\cap (\tS \setminus T)}}  \left(  \norm{\beta_G}/\sqrt{\abs{G}} + C'\sigma \sqrt{\frac{ \abs{S}\log(p)}{n \abs{G}}} \right)^{-1} \\
w_{\sco}^- &\geq \min_{G \in \mathcal{G}_{\cap S}^c}  \left(C' \sigma \sqrt{\frac{\abs{S} \log(p) }{n\abs{G}}}\right)^{-1}. \\
\intertext{With Condition B2, for all $C'' >0$, there exists $n_0$ so that for all $n>n_0$, with the same probability,}
w_T^+&\leq \max_{G \in \mathcal{G}_{\cap S}} \frac{\sqrt{\abs{G}}} { \norm{\beta_G} - C''\norm{\beta_G} \sqrt{\abs{G}^{-\gamma_1} n^{-\gamma_2}} } \\
w_{\tS \setminus T}^+&\geq \min_{G \in \mathcal{G}_{\cap (\tS \setminus T)}} \frac{\sqrt{\abs{G}}} { \norm{\beta_G} + C''\norm{\beta_G} \sqrt{\abs{G}^{-\gamma_1} n^{-\gamma_2}}} \\
w_{\sco}^- &\geq \frac{ \min_{G \in \mathcal{G}^c_{\cap S}}\sqrt{\abs{G}^{1+\gamma_1} n^{\gamma_2}}}{C'' \min_{H \in \mathcal{G}_{\cap S}} \norm{\beta_H}}.
\end{align*}

Assume then for the remainder of this proof that this is the case, together with the previous bounds on $\mu^{\varepsilon_{OLS}}$ and $\mu^{\varepsilon_{OLS}}_{\tS}$. We see that this is true with probability greater than $1- 3\exp(-t/2)/\sqrt{\pi}  \geq 1-\eta$.

Choosing $\mu = 2 \max(\mu^{\varepsilon_{OLS}}/w_{\sco}^-, \mu^{\varepsilon_{OLS}}_{\tS}/w_{\scit}^-)$ gives, by the definitions in Lemma~\ref{secondstage}, that it is simultaneously the case that for all $n \geq n_0$,

\begin{align*}
 L_1 &= \mu w_T^+ / \left(\mu w^-_{\sco} - \mu^{\varepsilon_{OLS}}\right) \\
&\leq 2\mu w_T^+ / \mu w^-_{\sco}  = 2w_T^+ / w^-_{\sco} \\
&\leq \max_{G \in \mathcal{G}_{\cap S}, H \in \mathcal{G}_{\cap S}^c} \frac{2\sqrt{\abs{G}}} { 1 - C'' \sqrt{\abs{G}^{-\gamma_1}n^{-\gamma_2}} } \frac{C'' } {\sqrt{\abs{H}^{1+\gamma_1} n^{\gamma_2} }} \\
 L_2 &= \mu w_T^+ / \left(\mu w^-_{\scit} - \mu^{\varepsilon_{OLS}}_{\tS}\right) \\
&\leq 2\mu w_T^+ / \mu w^-_{\scit}  = 2w_T^+ / w^-_{\scit} \\
&\leq \max_{G \in \mathcal{G}_{\cap S},  H \in \mathcal{G}_{\cap (\tS \setminus T)}} 2\frac{\norm{\beta_H} \sqrt{\abs{G}}}{\norm{\beta_G} \sqrt{\abs{H}}} \frac{1+C''\sqrt{\abs{H}^{-\gamma_1}n^{-\gamma_2}} }{1-C''\sqrt{\abs{G}^{-\gamma_1} n^{-\gamma_2}}}.
\end{align*}

Hence, for any value of $\delta > 0$, as $\abs{G}, n \rightarrow \infty$, eventually 
\[
L_1 \leq \max_{G \in \mathcal{G}_{\cap S}, H \in \mathcal{G}_{\cap S}^c} \delta\sqrt{\frac{\abs{G}}{\abs{H}^{1+\gamma_1}n^{\gamma_2}}}. 
\]
 Similarly, for all $\delta > 0$, with $\abs{G}, \abs{H}, n$ sufficiently large, it is the case that 
\[
L_2 \leq \max_{G \in \mathcal{G}_{\cap S},  H \in \mathcal{G}_{\cap (\tS \setminus T)}} 2(1+\delta) \frac{\norm{\beta_H} \sqrt{\abs{G}}}{\norm{\beta_G} \sqrt{\abs{H}}}.
\]

Therefore under condition C1(a) or condition C2(a), we have that there exists $C$ such that either $\phi(L_1, \tS)>C$ or $\phi(\max(L_1, L_2), T)>C$ . Hence by Lemma~\ref{secondstage},

\begin{align*}
 \norm{ X \left( \hbetaOT - \hbetaB\right)}_n &\leq 2\mu w_T^+\sqrt{\abs{T}}/C \\
 &\leq 2\sigma\sqrt{\frac{\abs{T}}{n}} \max \left(2\sqrt{\log(p)} \frac{w_T^+}{w_{\sco}^-}, 2\sqrt{\log{\abs{\tS}}} \frac{w_T^+}{w_{\scit}^-}\right) /C \\
 &= O\left(\sigma\sqrt{\frac{\abs{T}}{n}} \max \left(\sqrt{\log(p)}  \max_{G \in \mathcal{G}_{\cap S}, H \in \mathcal{G}_{\cap S}^c} \frac{\sqrt{\abs{G}}}{\sqrt{\abs{H}^{1+\gamma_1}n^{\gamma_2}}},\right. \right. \\
&\quad\left.\left. \sqrt{\log{\abs{\tS}}}  \max_{G \in \mathcal{G}_{\cap S},  H \in \mathcal{G}_{\cap (\tS \setminus T)}}  \frac{\norm{\beta_H} \sqrt{\abs{G}}}{\norm{\beta_G} \sqrt{\abs{H}}}\right)\right) \\
\norm{\hbetaOT_T - \hbetaB_T} &= O\left(\sigma\sqrt{\frac{\abs{T}}{n}} \max \left(\sqrt{\log(p)}  \max_{G \in \mathcal{G}_{\cap S}, H \in \mathcal{G}_{\cap S}^c} \frac{\sqrt{\abs{G}}}{\sqrt{\abs{H}^{1+\gamma_1}n^{\gamma_2}}},\right. \right. \\
&\quad\left.\left. \sqrt{\log{\abs{\tS}}}  \max_{G \in \mathcal{G}_{\cap S},  H \in \mathcal{G}_{\cap (\tS \setminus T)}}  \frac{\norm{\beta_H} \sqrt{\abs{G}}}{\norm{\beta_G} \sqrt{\abs{H}}}\right)\right).
\end{align*}

Similarly, under condition C1(b) or C2(b), we have that there exists $C$ such that either $\phi(L_1, \tS, \abs{\tS} + \abs{T})>C$ or $\phi(\max(L_1, L_2), T, 2\abs{T})>C$. Then, again by Lemma~\ref{secondstage},

\begin{align*}
\norm{\hbetaOT_T - \hbetaB_T} &= O\left(\sigma\sqrt{\frac{\abs{T}}{n}} \max \left(\sqrt{\log(p)}  \max_{G \in \mathcal{G}_{\cap S}, H \in \mathcal{G}_{\cap S}^c} \frac{\sqrt{\abs{G}}}{\sqrt{\abs{H}^{1+\gamma_1}n^{\gamma_2}}},\right. \right. \\
&\quad\left.\left. \sqrt{\log{\abs{\tS}}}  \max_{G \in \mathcal{G}_{\cap S},  H \in \mathcal{G}_{\cap (\tS \setminus T)}}  \frac{\norm{\beta_H} \sqrt{\abs{G}}}{\norm{\beta_G} \sqrt{\abs{H}}}\right)(1+ \max(L_1, L_2))\right).
\end{align*}

Now, under condition C1 or C2, the smallest singular value of $X_T/\sqrt{n}$ remains bounded away from zero. Therefore $\norm{X(\hbetaOT - \hbeta)}_n = O(\sigma \sqrt{\abs{T}/n})$, and $\norm{\hbetaOT - \hbeta} = O(\sigma \sqrt{\abs{T}/n})$, so using the triangle rule gives the result.

If $T = \tS$, we can proceed with the proof as normal, simply omitting the terms relating to $\scit$ and condition C2.

\end{proof}
\subsection{Proof of Corollary~\ref{corpropor} and Corollary~\ref{coraiao}}
\begin{proof}

 Simply apply Theorem~\ref{bigmain}, observing that given the condition on $\abs{G}, \abs{H}$, $L_1 \rightarrow 0$. In the case of Corollary~\ref{coraiao} condition C1(a) is thus implied by A1. We have then that in the second lemma,
\[
 \norm{X \left( \beta - \hbetaB\right)}^2_n = O\left( \sigma^2\frac{\abs{\tS}}{n}(1+ \sqrt{\log(p)\abs{G}^{-\gamma_1}n^{-\gamma_2}})^2 \right).
\]

But with non-overlapping groups with proportional sizes, $ p = O(\abs{\mathcal{G}}\abs{G})$, so 
\begin{align*}
\log(p) \abs{G}^{-\gamma_1}&= O(\log \mathcal{G}\abs{G}^{-\gamma_1} +  \log\abs{G}\abs{G}^{-\gamma_1}) \\
&= O(\log \mathcal{G}\abs{G}^{-\gamma_1} + 1).
\end{align*}

Corollary~\ref{corpropor} follows similarly.
\end{proof}

\bibliographystyle{apalike}
\bibliography{co-adaptivepaperT} 
\end{document}